\documentclass[a4paper,UKenglish]{lipics-v2016}
\usepackage[T1]{fontenc}
\usepackage{amsmath,amssymb,wasysym,dsfont}
\usepackage{amsthm}
\usepackage{mathtools,float}
\usepackage{xcolor,graphicx,units}
\usepackage{todonotes}
\usepackage{algorithm}
\usepackage[noend]{algpseudocode}
\usepackage{thmtools}
\usepackage{microtype}

\theoremstyle{plain}

\makeatletter
\newlength{\within@wd}
\DeclareRobustCommand{\within}[3][c]{%
  \ifmmode%
    \mathchoice%
      {\within@math\displaystyle{#1}{#2}{#3}}%
      {\within@math\textstyle{#1}{#2}{#3}}%
      {\within@math\scriptstyle{#1}{#2}{#3}}%
      {\within@math\scriptscriptstyle{#1}{#2}{#3}}%
  \else%
    \within@do{#1}{#2}{#3}%
  \fi%
}
\newcommand{\within@math}[4]{\within@do{#2}{$\m@th#1#3$}{$\m@th#1#4$}}
\newcommand{\within@do}[3]{%
  \settowidth{\within@wd}{#2}%
  \makebox[\within@wd][#1]{#3}%
}
\makeatother

\DeclareMathOperator{\dist}{{\mathit{d}}}

\DeclareMathOperator{\parent}{\operatorname{parent}}
\DeclareMathOperator{\Oh}{{O}}

\newcommand{\minR}{{\rho_{\mathrm{min}}}}
\newcommand{\maxR}{{\rho_{\mathrm{max}}}}
\newcommand{\Areas}{\mathcal{A}}

\newcommand{\Clients}{\mathcal{C}}

\newcommand{\Pairs}{{\Pi}}

\newcommand{\cround}[1]{{\ceil{\log_5{#1}}}}
\newcommand{\fround}[1]{{\floor{\log_5{#1}}}}

\DeclarePairedDelimiterX{\floor}[1]{\lfloor}{\rfloor}{#1}
\DeclarePairedDelimiterX{\ceil}[1]{\lceil}{\rceil}{#1}
\DeclarePairedDelimiterX{\prn}[1]{(}{)}{#1}
\DeclarePairedDelimiterX{\brc}[1]{\{}{\}}{#1}
\DeclarePairedDelimiterX{\brk}[1]{[}{]}{#1}
\DeclarePairedDelimiterX{\abs}[1]{\lvert}{\rvert}{#1}
\DeclarePairedDelimiterX{\tuple}[1]{\langle}{\rangle}{#1}
\DeclarePairedDelimiterX{\set}[2]{\{}{\}}{#1\,\delimsize|\,\mathopen{}#2}

\title{Dynamic clustering to minimize the sum of radii\footnote{The research leading to these results has received funding from the European Research Council
under the European Union's Seventh Framework Programme (FP/2007-2013) / ERC Grant
Agreement no. 340506.}}

\author[1]{Monika Henzinger}
\author[2]{Dariusz Leniowski}
\author[3]{Claire Mathieu}
\affil[1]{University of Vienna, Faculty of Computer Science, Vienna, Austria \\
  \texttt{monika.henzinger@univie.ac.at}}
\affil[2]{University of Vienna, Faculty of Computer Science, Vienna, Austria \\
  \texttt{dariusz.leniowski@univie.ac.at}}
\affil[3]{ENS, CNRS, {\tiny{PSL Research University}}, Paris, France \\
  \texttt{cmathieu@di.ens.fr}}
\authorrunning{M. Henzinger, D. Leniowski and C. Mathieu} 

\Copyright{Monika Henzinger, Dariusz Leniowski and Claire Mathieu}


\begin{document}
\maketitle
\begin{abstract}
  In this paper, we study the problem of opening centers to cluster 
a set of clients in a metric space so as to minimize the sum of 
the costs of the centers and of the cluster radii, in a dynamic 
environment where clients arrive and depart, and the solution must 
be updated efficiently while remaining competitive with respect 
to the current optimal solution. 
We call this \emph{dynamic sum-of-radii clustering} problem. 

We present a data structure that maintains a solution whose cost is within a constant factor of the cost of an optimal solution in metric spaces with bounded doubling dimension and whose
worst-case update time is logarithmic in the parameters of the problem.
\end{abstract}

\newcommand{\fmin}{f_{\mathrm{min}}}
\section{Introduction}

The main goal of clustering is to partition a set of objects into
homogeneous and well separated subsets (clusters). 
Clustering techniques have long been used in a wide 
variety of application areas, 
see for example the excellent surveys~\cite{schaeffer2007,Hansen1997}.

There are several ways to model clustering. 
Among them, the problem of sum-of-radii (or sum-of-diameter) clustering 
has been extensively studied: the clients are located in a metric space 
and one must open facilities to minimize facility opening cost 
(or keep the number of open facilities limited to at most $k$) 
plus the sum of the cluster radii (or, in other applications, 
cluster diameters).
To give a concrete example, imagine a telecommunications 
agency setting up mobile towers that provide wireless access
to selected clients, incurring costs for setting up towers as well as 
for configuring a tower to serve the customers lying within a certain 
distance, where that latter contribution to the cost increases with 
the maximum distance served by the tower. 

Assume the number of facilities is limited to $k$. 
For sum-of-diameter clustering, Doddi, Marathe, Ravi, 
Taylor and Widmayer~\cite{Doddi2000} prove hardness of approximation 
to within better than a factor of 2.  
More recently, NP-hardness was proved for the sum-of-radii problem 
even for shortest path metrics on weighted planar 
graphs~\cite{proietti2006}, or, in the case of sum-of-diameters, 
even for metrics of constant doubling dimension~\cite{gibson2008}. 
Turning to approximation algorithms, 
Charikar and Panigraphy~\cite{Charikar2004} 
design and analyze an $\Oh(1)$ approximation algorithm 
for sum-of-radii and for sum-of-diameter clustering with $k$ clusters.  
They start from a linear-programming relaxation, 
use a primal-dual type approach and, along the way, 
design a bicriteria algorithm. 
They also design an incremental algorithm that handles arrivals of clients,
merging clusters as needed so that at any time the clustering has 
$O(k)$ clusters and the cost is $\Oh(1)$ times the optimal cost for 
$k$ clusters.

There have been many other papers on sum-of-radii or sum-of-diameters 
clustering. A few papers focus on the problems of partitioning 
the clients into a constant number of clusters as quickly as 
possible~\cite{hansen1987,proietti2006}. 
Some papers concern themselves with bicriteria results such 
as~\cite{Bandyapadhyay2015}.
Consider the special case of a metric that is Euclidean in two dimensions. 
Lev-Tov and Peleg~\cite{Lev-Tov2005} give a polynomial time approximation 
scheme (PTAS) for the related problem of covering input clients by 
a min-cost set of disks centered at the servers, 
where both clients and potential servers are located in 
the Euclidean plane and are part of the input. 
Recently, Behsaz and Salavatipour~\cite{Behsaz2015} gave a PTAS 
for the minimum sum of diameters problem on the plane with Euclidean 
distances. 
See also~\cite{Alt2006,Gibson2008b} for other work on the 
two-dimensional geometric setting.

The problems are complicated by situations where the set of clients 
may change over time, 
for example documents in a very large database 
that must be efficiently searchable and maintained.
This then leads to various models: 
online~\cite{Csirik2013,Fotakis2011}, incremental\cite{Charikar2004}, 
streaming, or dynamic. 
The dynamic setting, where clients may not only arrive but also depart, 
has been empirically studied at least since 1993~\cite{can1993}, 
and is the focus of the present paper, 
with the joint goals of maintaining clusterings whose objective value 
is close to optimal, and of updating the cluster quickly after each event.

This paper can be interpreted as part of a recent focus on exploiting 
primal-dual techniques 
in the dynamic setting. 
In the online setting (where new elements arrive but never depart), 
primal-dual techniques are extremely successful~\cite{Buchbinder2009}. 
Initially it seemed that such techniques 
were inherently restricted to settings with arrivals only, 
and no departures, 
but there recently has been exciting progress 
to handle the dynamic setting as well, 
starting with~\cite{BhattacharyaHI15,Bhattacharya2015} and continuing 
with~\cite{Baswana2015,Bhattacharya2016,Solomon2016,Bhattacharya2017}. 
Of particular notice for us is recent work by 
Gupta et al~\cite{gupta2016} 
for the set cover problem (and Bhattacharya et al~\cite{BhattacharyaCH16},
restricted to vertex cover) in the dynamic setting where elements 
arrive and depart.%

In this paper, we study the {\em dynamic sum-of-radii clustering problem}, 
defined as follows:
The original input consists of a (possibly infinite) set $V$ of potential 
{\em clients} or {\em points}, 
a finite set $F\subseteq V$ of {\em facilities} with an {\em opening cost} 
$f_j$ for each facility $j$, and a metric $d$ over $V$. 
For the online input, a set $C$ of live clients evolves over time: 
at each timestep $t$, either a new client arrives and is added to $C$, 
or a client from $C$ departs and is removed from $C$, 
or a query is made for the approximate cost of an optimal solution 
({\em cost query}), or a query asks for the entire current solution 
({\em solution query}). 
For the output, at each timestep the algorithm maintains 
a set of {\em open} facilities, 
each open facility $j$ being associated to a radius $R_j$, 
such that every client of $C$ is {\em covered}, i.e. 
belongs to some open ball $B(j,R_j)$, 
and the goal is to minimize the \emph{cost}, namely, 
the sum over open facilities $j$ of $f_j+R_j$.

The { dynamic sum-of-radii clustering problem} 
can actually be interpreted as a special case of dynamic set cover: 
our metric space is the universe, our clients are the elements, 
and for each center $c$ and radius $r$, 
the ball $B(c,r)$ defines a set of cost $r$ 
consisting of those clients covered by the ball. 
The dynamic set cover algorithm from~\cite{gupta2016}, 
specialized to our setting, maintains competitive ratio 
$\Oh(\log n)$ and has update time $\Oh(f\log n)$, 
where $f$ is the maximum number of sets containing any element; 
\cite{gupta2016} also gives an $\Oh(f^3)$ 
approximation in time $\Oh(f^2)$ for set cover, and 
for the related dynamic $k$-coverage problem, 
they give a constant approximation fully dynamic algorithm with 
$\Oh(f \log n)$ update time. 

The {\em doubling dimension} of a metric space $(V,d)$ is said 
to be bounded by $\kappa$ if any ball $B(x,r)$ in $(V,d)$ 
can be covered by $2^\kappa$ balls of radius $r/2$~\cite{Krauthgamer2004}. 
For example, $D$-dimensional Euclidean space has doubling dimension 
$\Theta(D)$.

In this paper we give an algorithm for constant doubling dimension, 
that maintains a solution 
whose cost is within a constant factor of the optimal cost, 
and with logarithmic update times for arrival 
or departure events. 
The algorithm answers cost queries in constant time and solution 
queries in linear time in the size of the solution, 
up to a factor of $\log (W/f_{min})$, where $W$ is the diameter of 
the metric space and 
$f_{min}$ is the minimum opening cost of any facility (which is strictly positive without loss of generality, since facilities of cost 0 can remain open at all times).
The universe is known ahead of time, 
and only the collection of active clients changes dynamically. 
Note that for the above mentioned algorithm by~\cite{gupta2016} 
$f$ would be $\Theta(2^{2 \kappa} \log (W/f_{min}))$.

\begin{theorem}\label{thm:main}
There exists an algorithm for the dynamic sum-of-radii clustering problem, 
when clients and facilities live in a metric space with doubling dimension 
$\kappa$, such that at every timestep the solution has cost at most 
$\Oh(2^{2 \kappa})$ times the cost of an optimal solution at that time, 
and such that the update time is $\Oh(2^{6 \kappa} \log (W/f_{min}))$,
where $W$ is the diameter of the space and in the current number of 
clients and $f_{min}$ is the minimum opening cost.
A cost query can be answered in constant time, a solution query in time 
$\Oh( s \log (W/f_{min}))$, where $s$ is   the size of the output.
\end{theorem}

The $2^{6\kappa}$ factor in the update time is due to 
fixed-radius nearest neighbor query that we solve using 
only the basic data structure already present in the algorithm.
However, it can be improved -- for example in case of finite metric spaces 
where a $\Oh(1)$ lookup table over the space is possible
(e.g., graphs with the shortest path distance), 
the update time reduces to $\Oh(\log(W/f_{\mathrm{min}}))$
at the cost of additional preprocessing time.
More generally, if the metric space allows to answer fixed-radius nearest neighbor queries 
in time $\Oh(d)$, then the update time becomes $\Oh(d \log(W/f_{\mathrm{min}}))$,
at the cost of preprocessing time necessary to construct the oracle.

We show the following structural property for the sum-of-radii clustering problem (Theorem~\ref{theorem:areas-structure}): 
there exists a collection $\Pairs$ of pairs $\tuple{j,r}$ 
where $j\in F$ and $r$ is a non-negative integer, 
each with an associated area $A(j,r)$ of $V$, 
and an abstract tree $\cal T$ over $\Pairs$, with the following properties:
\begin{enumerate}
\item
$\mathcal{T}$ has height $\Oh (\log (W/\fmin))$ and degree at most $2^{4\kappa}$.
\item
The collection $\mathcal{A}$ of areas is a laminar family,  its laminar structure is given by $\mathcal T$, and for each area, $A(j,r)\subseteq B(j,7\cdot 5^r)$
\item
For any subset $\Clients$ of $V$,  there exists a collection $\cal S$ of areas  covering $\Clients$ and whose cost, $\sum_{\tuple{j,r}\in{\mathcal{S}}} f_j+7\cdot 5^r $, is $O(2^{2\kappa})$ times the optimal cost for $C$. 
\end{enumerate}

Our algorithm has two phases.
First, in the preprocessing phase (Section~\ref{sc:preprocessing}), 
the algorithm constructs $\Pairs$, 
the laminar family of areas $\mathcal{A}$ and corresponding abstract tree $\mathcal{T}$. 
Thanks to the last property above, it suffices to restrict attention
to solutions that use only areas $A(j,r)$ for coverage, with  $\tuple{j,r}\in \Pairs$. 
Second, in the dynamic phase (Section~\ref{sc:algorithm}), 
while clients arrive and depart, the algorithm maintains an optimal 
set of pairs $\tuple{j,r}$ of $\Pairs$ such that the corresponding areas $A(j,r)$ 
cover all current clients. 
The hierarchical structure of $\mathcal{T}$ makes this simple, so that each 
update takes time proportional to the height times $2^{6\kappa}$.

The main contribution of the paper is the definition of $\Pairs$ and 
the corresponding laminar family of areas $\mathcal{A}$.
The latter is reminiscent of the cover tree data structure of \cite{CoverTrees06}.
However, the cover tree is tailored to the nearest neighbor problem
and its covers, to the best of our knowledge, lack the structural properties of areas
that we need to prove the approximation factor for the sum-of-radii clustering problem.
We expect that our new structure can be used for other clustering-type problems.

\section{Preprocessing phase}\label{sc:preprocessing}

\subsection{Discretization of radii}

Given the set of clients $\Clients\subseteq V$, let $OPT$ denote the cost of an optimum solution for $\Clients$.

\begin{lemma}\label{lemma:preprocessing}
For all $\Clients\subseteq V$, there exists a solution such that every ball $B(j,R)$ has $\fmin\leq f_j\leq R\leq 5\cdot W$, the radius $R$ is an integer power of 5, and the cost is $O(OPT)$.
\end{lemma}

\begin{proof}
Consider the unknown optimal solution. 
If some ball is such that $\max(f_j,R)>W$ then replace the entire solution by a ball centered at the facility of cost $\fmin$ and of radius $W$. Else, for each ball $B(j,R)$ of the optimal solution:
\begin{itemize}
\item if $f_j>R$ then increase the radius of the ball from $R$ to $f_j$
\item Increase $R$ to the smallest integer power of 5 that is greater than or equal to $R$.
\end{itemize}
The new solution satisfies the desired constraints, and the cost has increased by a factor of 10 at most.
\end{proof}

A \emph{logradius} is an integer $r$ such that $\fmin \leq 5^r \leq 5\cdot W$. Let $\minR = \fround{f_{\min}}$ and $\maxR = \cround{W}$.
Then the number of different logradii, $\maxR - \minR + 1$,
is $\Oh(\log ({W}/{f_{\min}}))$. 

\subsection{Maximal subsets of distant facilities}
We construct a set $\Pairs$ of {\emph pairs} $\tuple{j,r}$ where $j$ is a facility and $r$ is a logradius,  satisfying the following properties:
\begin{enumerate}
\item (Covering) For every facility $j\in J$ and every logradius $r$ such that $f_j\leq 5^r$, there exists a facility $j'\in J_{r}$ with $\dist(j,j')\leq  5^{r+1}$ and $\tuple{j',r}\in \Pairs$.
\item (Separating) For all distinct $\tuple{j',r}, \tuple{j',r}\in \Pairs$, we have $\dist(j,j')> 5^{r+1}$.
\end{enumerate}

\begin{center}
\fbox{\parbox{\textwidth}{
For each logradius $r \in [\minR, \maxR]$:
\begin{itemize}
\item let $J_r' = \set{j \in F}{f_j \leq 5^r}$. 
\item let $J_r$ be a maximal subset of $J_r'$ 
such that any two facilities in $J_r$ are at distance greater than $5^{r+1}$. 
\end{itemize}
$\Pairs \gets \bigcup_r \set{\tuple{j,r}}{j \in J_r}$.}}
\end{center}

Note that for $r=\maxR$, the set $J_r$ contains just one facility.

\subsection{Hierarchical decomposition of $\Pairs$}

\begin{center}
\fbox{\parbox{\textwidth}{
Construct an abstract tree $\mathcal{T}$ over $\Pairs$ as follows (with ties broken arbitrarily):
\begin{itemize}
\item
the root of $\mathcal{T}$ is the unique pair $\tuple{j,\maxR}$.
\item
for all $r<\maxR$ and $j\in J_r$:
	\begin{itemize}
	\item let $j'$ be the facility of $J_{r+1}$ closest to $j$ 
	\item  $\parent(j,r)\gets \tuple{j',r+1}$
	\end{itemize}
\end{itemize}
}}
\end{center}

By construction, $\mathcal{T}$ has height at most $\maxR-\minR+1$ and the parent of a pair $\tuple{j,r}$ is a pair of the form $\tuple{j',r+1}$.

The following Lemma is simple, but it captures the essential way in which using larger balls will greatly simplify the structure, and is the main step towards constructing a laminar set of areas for covering clients.  
\begin{lemma}\label{lemma:inclusion-of-big-balls}
(Nesting of balls) If $\parent(j,r)=\tuple{j',r+1}$, then $B(j,7\cdot 5^r)\subseteq B(j',7\cdot 5^{r+1})$.
\end{lemma}
\begin{proof}
We have $\tuple{j,r}\in \Pairs$, so $j\in J'_r\subseteq J'_{r+1}$. By the Covering property of $\Pairs$ the maximum distance from any point in $B(j,7\cdot 5^r)$ to $j'$ is $d(j,j')+7\cdot 5^r\leq 5^{r+2}+7\cdot 5^r\leq 7\cdot 5^{r+1}$.
  \end{proof}

\begin{lemma}\label{lem:doubling_dimension}
For any point $p$ and radius $r$, the set of pairs
\[\Pairs(p,r) = \set{\tuple{j,r} \in \Pairs}{\dist(p,j) < 2^{\alpha} 5^{r+1}}\]
has at most $2^{(\alpha +1)\kappa}$ elements, where $\kappa$ is the doubling
dimension of the metric space.
\end{lemma}
\begin{proof}
By definition of doubling dimension, $B(p, 2^{\alpha}\cdot 5^{r+1})$ can be
covered by a set of at most $(2^{\kappa})^{\alpha+1}$ balls of radius $(1/2)\cdot 5^{r+1} $.
By the Separating property of $\Pairs$, any two pairs $\tuple{j,r}$ of $\Pairs(p,r)$ are at distance greater than $5^{r+1}$ from each other, hence must belong to  different balls of the set, and so $\Pairs(p,r)$ has cardinality at most $(2^{\kappa})^{\alpha+1}$. 
\end{proof}

\begin{lemma}\label{lemma:degree}
A node $\tuple{j,r}$ of $\mathcal{T}$ has at most $2^{4\kappa}$ children
\end{lemma}
\begin{proof}
Children of  $\tuple{j,r}$
have logradius $r-1$, so 
by the Covering property of $\Pairs$ their distance to $j$ is at most $5^{r+1}$, so they belong to $\Pairs(j,r-1)$ for $\alpha=3$, and so Lemma~\ref{lem:doubling_dimension} applies.
\end{proof}

\subsection{Hierarchical decomposition of $V$ into a laminar family of areas}\label{sec:decomposition_into_areas}
Recall that a collection $\Areas$ of sets is  \emph{laminar} if for any two $A,B\in \Areas$, either $A\cap B=\emptyset$ or $A\subseteq B$ or $B\subseteq A$. We partition $V$ into a laminar family of \emph{areas},
denoted by $\Areas$,
such that no two same-logradius areas overlap.

\begin{center}
\fbox{\parbox{\textwidth}{
For each $\tuple{j,r}\in\Pairs$, initialize $A(j,r)\gets\emptyset$.\\
For each point $p\in V$:
\begin{itemize}
\item let $r^*$ be minimum such that there exists pairs $\tuple{j,r^*}$ with $p\in B(j,7\cdot 5^{r^*})$. 
\item Among all such pairs, let $\tuple{j^*,r^*}$ denote the one minimizing $\dist (p,j^*)$. 
\item Add $p$ to the set $A(j^*,r^*)$ and to every set $A(j',r')$ with $(j',r')$ ancestor of $(j^*,r^*)$ in $\mathcal{T}$. 
\end{itemize}
}}
\end{center}

\begin{lemma}\label{lemma:area-contained-in-big-ball}
For every $\tuple{j,r} \in \Pairs$, $A(j,r)\subseteq B(j,7\cdot 5^r).$
\end{lemma}
\begin{proof}
Let $p\in A(j,r)$. Either it's been added directly, in which case it belongs to $B(j,7\cdot 5^r)$, or it's been inherited, in which case it also belongs to it by Lemma~\ref{lemma:inclusion-of-big-balls}.
\end{proof}

\begin{lemma}\label{lemma:Areas-suffice}
For every subset $\mathcal{C} \subseteq V$ of clients there exists  $S\subseteq\Pairs$ such that $\mathcal{C}$ is covered by $\cup \{ A(j,r): \tuple{j,r}\in S \}$   and 
$\sum_{\tuple{j,r}\in S} (f_j+7\cdot 5^r)=\Oh(2^{2\kappa}\cdot OPT)$.
\end{lemma}
\begin{proof}
Let $S^*$ be a solution of cost $\Oh(OPT)$ satisfying the properties of Lemma~\ref{lemma:preprocessing}. For each ball $B(j,5^r)$ of $S^*$, put in $S$ all the pairs $\tuple{j',r}\in\Pairs$ such that $\dist(j,j')\leq 8\cdot 5^r$. 

We claim that  $\mathcal{C}$ is covered by $\cup \{ A(j',r): \tuple{j',r}\in S \}$. Indeed, consider a client $p\in\mathcal{C}$ and a ball $B(j,5^r)$ of $S^*$ containing $p$. By the Covering property of $\Pairs$, there exists $\tuple{j',r}\in\Pairs$ with $\dist(j,j')\leq 5^{r+1}$. Then $\dist(p,j')\leq 5^{r+1}+5^r < 7 \cdot 5^r$, and so in the definition of areas covering $p$ we must have $r^*\leq r$. Along the path from $\tuple{j^*,r^*}$ to the root of   $\mathcal{T}$, there exists a pair  for logradius $r$, $\tuple{j'',r}$. By definition of areas and by Lemma~\ref{lemma:area-contained-in-big-ball}, $p\in A(j'',r)\subseteq B(j'',7\cdot 5^r)$,  so $\dist(j,j'')\leq \dist(j,p)+\dist(p,j'')\leq 8\cdot 5^r$, and therefore $\tuple{j'',r}\in S$ and $p$ is covered. 

In terms of costs, since all these areas are associated to pairs within distance $8\cdot 5^r< 2\cdot  5^{r+1}$ from $j$, by Lemma~\ref{lem:doubling_dimension} for $\alpha=1$, there are at most $2^{2\kappa}$ of them.
\end{proof}

By Lemma~\ref{lemma:area-contained-in-big-ball} and the definition of areas, we note that $\cup_{J_r}A(j,r)=\cup_{J_r}B(j,7\cdot 5^r)$, so we also give hereafter an equivalent description of the same laminar family, illustrating the way in which the parent-child relations in tree $\mathcal{T}$ and the proximity relations in the metric space are balanced against one another. (This also has the advantage of being constructive even if $V$ is infinite).

\noindent Partition $\cup_{J_{\minR}} B(j,7\cdot 5^{\minR})$,  using the facilities of $J_\minR$ as centers, into Voronoi cells $A(j,\minR)$. \\
For $r\in (\minR,\maxR]$:
\begin{itemize}
\item
Partition $\cup_{j\in J_r} B(j,7\cdot 5^{r})\setminus \cup_{j\in J_{r-1}} B(j,7\cdot 5^{r-1})$, using the facilities of $J_r$ as centers, into Voronoi cells $A(j,r)$.
\item
For each $\tuple{j,r}\in\Pairs$, $A(j,r)\gets A(j,r)\cup\bigcup \{ A(j',r-1): \parent(j',r-1)=\tuple{j,r} \}$
\end{itemize}

The construction of this section can be summarized in the following structural Theorem.

\begin{theorem}\label{theorem:areas-structure}
Let a metric space $(V,d)$ of doubling dimension $\kappa$ be given, as well as a subset $F$ of elements of $V$ called {\emph facilities}, with an associated {\emph cost} $f_j$ for each $j\in F$. Then there exists an abstract tree $\mathcal{T}$ whose nodes are indexed by facilities $j\in F$ and non-negative integers $r\geq 0$, and, for each node $\tuple{j,r}$, an associated {\emph area} $A(j,r)\subseteq V$ with the following properties
\begin{enumerate}
\item
$\mathcal{T}$ has height $\Oh (\log (W/\fmin))$ and degree at most $2^{4\kappa}$, and for each $\tuple{j,r}\in {\mathcal{T}}$ and its parent node $\tuple{j',r+1}$, $B(j,7\cdot 5^r)\subseteq B(j',7\cdot 5^{r+1})$.
\item
$\mathcal{A}$ is a laminar family,  its laminar structure is given by $\mathcal T$, and for each area $A(j,r)$, $A(j,r)\subseteq B(j,7\cdot 5^r)$
\item
For any subset $\Clients$ of $V$, for any collection of balls $\cal B$ centered at facilities of $F$ and covering $\Clients$, there exists a collection $\cal S$ of areas  covering $\Clients$, such that $\sum_{\tuple{j,r}\in{\mathcal{S}}} f_j+7\cdot 5^r = 
O(2^{2\kappa}) \sum_{B(j,R)\in{\mathcal{B}}}(f_j+R)$.
\end{enumerate}
\end{theorem}

\newcommand{\Tau}{{\mathcal{T}}}
\section{Data structure}\label{sc:algorithm}

\subsection{Solving the offline restricted problem}\label{section:offline}

Given $\Clients\subseteq V$, we wish to compute the solution of minimum cost  among all {\em solutions that are restricted to covering $\Clients$ using areas $A(j,r)$ for $\tuple{j,r}\in \Pairs$}, where using area $A(j,r)$ has cost $f_j+c_2\cdot 5^r$. We call that the {\em restricted} problem.
By Theorem~\ref{theorem:areas-structure}  the optimal restricted cost is a $O(2^{2\kappa})$ approximation of the optimal (unrestricted) cost.

Computing the optimal solution to the restricted problem in an offline manner is straightforward, thanks to the laminar structure of the candidate areas. We first compute, for each node $\tuple{j,r}$ of $\mathcal{T}$, the cost $c_{j,r}=f_j+c_2\cdot 5^r$ of area $A(j,r)$, as well as the number $n_{j,r}$ of clients that are in area $A(j,r)$ but not in any of the areas of children nodes: since areas $A(j',r-1)$ are all disjoint by laminarity, we have 
$n_{j,r}= |\Clients \cap A(j,r)| - \cup_{\tuple{j',r-1}: \parent(j',r-1)=j} | \Clients \cap A(j',r-1)| $. We  then compute  the optimal cost $x_{j,r}$ of covering the clients of $\Clients\cap A(j,r)$ using only areas of the subtree of $\mathcal{T}$ rooted at $\tuple{j,r}$, using the following bottom-up recurrence:

\begin{center}
\fbox{\parbox{\textwidth}{
For $\tuple{j,r}\in \Pairs$ in bottom-up order in $\mathcal{T}$:\\
$$x_{j,r}=\begin{cases}
c_{j,r}  \text{ if } n_{j,r}>0 \\
\min\left( c_{j,r},\sum\{ x_{j',r-1} : \tuple{j,r}= \parent({j',r-1})\} \right)  \text{ otherwise.}
\end{cases}$$
}}
\end{center}

Indeed,  if $n_{j,r}\neq 0$ then  the solution must use area $A(j,r)$; but then by laminarity area $A(j,r)$ covers all clients in that subtree, so no other area is needed in the solution, and the cost is exactly the cost $c_{j,r}$ of $A(j,r)$. If on the other hand $n_{j,r}=0$, then we have an alternative possibility: we could do without $A(j,r)$. Then, by disjointness of sibling areas the problem separates into independent subproblems, one for each child of $\tuple{j,r}$, hence the recurrence simply sums their costs.

The cost of the optimal restricted solution is then $x_{j,\maxR}$ for the root $\tuple{j,\maxR}$ of $\mathcal{T}$. 

Given $c_{j,r}$ and $x_{j,r}$, computing the optimal restricted solution, a collection $S$ of areas, is done recursively:
\begin{center}
\fbox{\parbox{\textwidth}{
$$S(j,r)=\begin{cases}
\emptyset \text{ if } x_{j,r}=0\\
\{A(j,r)\} \text{ if } x_{j,r}=c_{j,r}\\
\cup \{ S(j',r-1) : \parent(j',r-1)=\tuple{j,r} \} \text{ otherwise.}
\end{cases}$$
}}
\end{center}
Thus the algorithm to compute the optimal set $S$ of areas covering $\Clients$ in the restricted problem, given the values of $c_{j,r},x_{j,r}$ explores a tree $\mathcal{T'}$ that, as it is a partial subtree of $\mathcal{T}$, also has height at most $\Oh(\log(W/\fmin))$ and degree at most $2^{4\kappa}$; moreover its  internal nodes are all ancestors of areas added to the solution $S$, so the running time to compute $S$ itself is $\Oh(2^{4\kappa}\log(W/\fmin) |S|  )$.

\subsection{The dynamic data structure}

The dynamic data structure supports insertions of clients, deletions of clients, queries for  the cost of the optimal restricted solution, and  queries for  the set of open facilities and areas of the optimal restricted solution.

The algorithm will maintain two dynamic data structures:
\begin{enumerate}
\item
a list of the currently existing clients $\Clients\subseteq V$, with, for each client $p$, the $\tuple{j,r}\in\Pairs$ such that $p\in A(j,r)$ and $r$ is minimum; and
\item an {\em annotated dependency tree} $\Tau_A$, keeping for each node $v = \tuple{j,r}$ the following additional information:
	\begin{enumerate}
	\item
 	its cost $c_v = f_j + 7\cdot 5^r$, 
 	\item 
	 the number $n_v$ of currently existing clients that belong to $A(j,r)$  but not to any  descendant area,
	 \item 
	  the value $x_v$, which is the minimum cost needed to cover all clients belonging to $A(j,r)$ using only areas $A(j',r')$ for $\tuple{j',r'}\in\Pairs$,  and
	  \item the value $y_v = \sum_{u \textit{ child of } v} x_u$.
	  \end{enumerate}
\end{enumerate}

To initialize the data structures, from the preprocessing phase  the algorithm is given the set $\Pairs$ of pairs $\tuple{j,r}$, as well as the laminar family of areas $\Areas$ with its dependency tree $\mathcal{T}$ using the following representation, which can be easily computed in time linear in its size:
(1)  An array of size $\maxR - \minR + 1$, keeping for each logradius $r\in [\minR,\maxR]$ a list of all the facilities of $J_r$, and
(2)  An annotated tree data structure obtained from $\mathcal{T}$ by setting every $n_v,x_v,y_v$ equal to 0, and $c_{j,r}=f_j+7\cdot 5^r$. The initial set of clients is $\Clients=\emptyset$. 

Answering queries is done as in Section~\ref{section:offline}.

We next describe the client deletions. 
When a client $p$ is deleted, we start from $\tuple{j,r}$ in $\Tau_A$, such that $p\in A(j,r)$ and $r$ is minimum; we decrement $n_v$ and we traverse the path from $\tuple{j,r}$ up to the root of $\Tau_A$, updating $ x_v$ and $y_{\parent(v)}$ for every node visited along the way using the recurrence from Section~\ref{section:offline}. This takes time proportional to the height of the tree, $\Oh(\log(W/\fmin))$.

Similarly, when a client $p$ is inserted, we first find $\tuple{j,r}$ in $\Tau_A$, such that $p\in A(j,r)$ and $r$ is minimum, in a way to be described shortly; we increment $n_v$, and then we traverse the path from $\tuple{j,r}$ up to the root of $\Tau_A$, similarly  updating $x_v$ and $y_{\parent(v)}$.

Thus, it only remains to determine the pair $\tuple{j^*,r^*}$ with smallest logradius such that $p\in A(j^*,r^*)$. By Lemma~\ref{lemma:area-contained-in-big-ball}, 
$p\in B(j^*,7\cdot 5^{r^*}).$ Thus we will first find all pairs $\tuple{j,r}$ such that $p\in B(j,7\cdot 5^{r})$, based on them determine $r^*$, and then look for $\tuple{j^*,r^*}$ in that set of balls.  
Thanks to Lemma~\ref{lemma:inclusion-of-big-balls}, the first part can be done using a simple recursive algorithm starting from the root of $\Tau_A$ (see below). The second part simply uses the definition of areas,
i.e., it finds the pair $\tuple{j^*,r^*}$ where $j^*$ has with minimum distance to $p$ out of all pairs $\tuple{j,r^*}$ with$p \in B(j,7\cdot 5^{r^*}).$

\begin{center}
\fbox{\parbox{\textwidth}{
$$\text{Pairs}(p,j,r)=\begin{cases}
\emptyset \text{ if } p\notin B(j,7\cdot 5^{r})\\
\{ \tuple{j,r}\} \cup \bigcup \{ \text{Pairs}(p,j',r-1) : \parent(j',r-1)=\tuple{j,r} \} \text{ otherwise.}
\end{cases}$$
\begin{itemize}
\item let $r^*$ be minimum such that there exists pairs $\tuple{j,r^*}$ in the set $\text{Pairs}(p,j_{\text{root}},\maxR)$. 
\item Among all such pairs, output the pair $\tuple{j^*,r^*}$ minimizing $\dist (p,j^*)$. 
\end{itemize}
}}
\end{center}

The running time is dominated by the first part, which is $\Oh(2^{4\kappa})$ times the number of pairs $\tuple{j,r}$ such that $p\in B(j,7\cdot 5^{r})$.
There are $\log(W/\fmin)$ possible values of $r$. 
For each $r$, 
 by Lemma~\ref{lem:doubling_dimension} there are at most  $2^{2\kappa}$ pairs  $\tuple{j,r}\in \Pairs$ such that $p\in B(j,c_2\cdot 5^{r})$ and the
algorithm has to test the $\Oh(2^{4 \kappa})$ children of each of them. 
Thus the running time to do an insertion is $\Oh(2^{6 \kappa} \log(W / f_{min})) $.

\clearpage
\bibliographystyle{plainurl}
\bibliography{main}

\begin{thebibliography}{10}

\bibitem{Alt2006}
Helmut Alt, Esther~M. Arkin, Herv{\'e} Br\"{o}nnimann, Jeff Erickson,
  S\'{a}ndor~P. Fekete, Christian Knauer, Jonathan Lenchner, Joseph S.~B.
  Mitchell, and Kim Whittlesey.
\newblock Minimum-cost coverage of point sets by disks.
\newblock In {\em Proceedings of the Twenty-second Annual Symposium on
  Computational Geometry}, SCG '06, pages 449--458, New York, NY, USA, 2006.
  ACM.
\newblock \href {http://dx.doi.org/10.1145/1137856.1137922}
  {\path{doi:10.1145/1137856.1137922}}.

\bibitem{Bandyapadhyay2015}
Sayan Bandyapadhyay and Kasturi~R. Varadarajan.
\newblock Approximate clustering via metric partitioning.
\newblock {\em CoRR}, abs/1507.02222, 2015.
\newblock URL: \url{http://arxiv.org/abs/1507.02222}.

\bibitem{Baswana2015}
Surender Baswana, Manoj Gupta, and Sandeep Sen.
\newblock Fully dynamic maximal matching in o(log n) update time.
\newblock {\em {SIAM} J. Comput.}, 44(1):88--113, 2015.
\newblock \href {http://dx.doi.org/10.1137/130914140}
  {\path{doi:10.1137/130914140}}.

\bibitem{Behsaz2015}
Babak Behsaz and Mohammad~R. Salavatipour.
\newblock On minimum sum of radii and diameters clustering.
\newblock {\em Algorithmica}, 73(1):143--165, September 2015.
\newblock \href {http://dx.doi.org/10.1007/s00453-014-9907-3}
  {\path{doi:10.1007/s00453-014-9907-3}}.

\bibitem{CoverTrees06}
Alina Beygelzimer, Sham Kakade, and John Langford.
\newblock Cover trees for nearest neighbor.
\newblock In William~W. Cohen and Andrew Moore, editors, {\em Machine Learning,
  Proceedings of the Twenty-Third International Conference {(ICML} 2006),
  Pittsburgh, Pennsylvania, USA, June 25-29, 2006}, volume 148 of {\em {ACM}
  International Conference Proceeding Series}, pages 97--104. {ACM}, 2006.
\newblock \href {http://dx.doi.org/10.1145/1143844.1143857}
  {\path{doi:10.1145/1143844.1143857}}.

\bibitem{BhattacharyaCH16}
Sayan Bhattacharya, Deeparnab Chakrabarty, and Monika Henzinger.
\newblock Deterministic fully dynamic approximate vertex cover and fractional
  matching in {\textdollar}o(1){\textdollar} amortized update time.
\newblock {\em CoRR}, abs/1611.00198, 2016.
\newblock URL: \url{http://arxiv.org/abs/1611.00198}.

\bibitem{Bhattacharya2015}
Sayan Bhattacharya, Monika Henzinger, and Giuseppe~F. Italiano.
\newblock Design of dynamic algorithms via primal-dual method.
\newblock In {\em Automata, Languages, and Programming - 42nd International
  Colloquium, {ICALP} 2015, Kyoto, Japan, July 6-10, 2015, Proceedings, Part
  {I}}, pages 206--218, 2015.
\newblock URL: \url{http://dx.doi.org/10.1007/978-3-662-47672-7_17}.

\bibitem{BhattacharyaHI15}
Sayan Bhattacharya, Monika Henzinger, and Giuseppe~F. Italiano.
\newblock Deterministic fully dynamic data structures for vertex cover and
  matching.
\newblock In {\em Proceedings of the Twenty-Sixth Annual {ACM-SIAM} Symposium
  on Discrete Algorithms, {SODA} 2015, San Diego, CA, USA, January 4-6, 2015},
  pages 785--804, 2015.
\newblock \href {http://dx.doi.org/10.1137/1.9781611973730.54}
  {\path{doi:10.1137/1.9781611973730.54}}.

\bibitem{Bhattacharya2016}
Sayan Bhattacharya, Monika Henzinger, and Danupon Nanongkai.
\newblock New deterministic approximation algorithms for fully dynamic
  matching.
\newblock In {\em Proceedings of the 48th Annual {ACM} {SIGACT} Symposium on
  Theory of Computing, {STOC} 2016, Cambridge, MA, USA, June 18-21, 2016},
  pages 398--411, 2016.
\newblock URL: \url{http://doi.acm.org/10.1145/2897518.2897568}, \href
  {http://dx.doi.org/10.1145/2897518.2897568}
  {\path{doi:10.1145/2897518.2897568}}.

\bibitem{Bhattacharya2017}
Sayan Bhattacharya, Monika Henzinger, and Danupon Nanongkai.
\newblock Fully dynamic approximate maximum matching and minimum vertex cover
  in $o(\log^3 n)$ worst case update time.
\newblock In {\em Proceedings of the Twenty-Eighth Annual {ACM-SIAM} Symposium
  on Discrete Algorithms, {SODA} 2017, Barcelona, Spain, Hotel Porta Fira,
  January 16-19}, pages 470--489, 2017.
\newblock \href {http://dx.doi.org/10.1137/1.9781611974782.30}
  {\path{doi:10.1137/1.9781611974782.30}}.

\bibitem{Buchbinder2009}
Niv Buchbinder and Joseph (Seffi)~Naor.
\newblock The design of competitive online algorithms via a primal: Dual
  approach.
\newblock {\em Found. Trends Theor. Comput. Sci.}, 3(2\&\#8211;3):93--263,
  February 2009.
\newblock \href {http://dx.doi.org/10.1561/0400000024}
  {\path{doi:10.1561/0400000024}}.

\bibitem{can1993}
Fazli Can.
\newblock Incremental clustering for dynamic information processing.
\newblock {\em ACM Trans. Inf. Syst.}, 11(2):143--164, April 1993.
\newblock \href {http://dx.doi.org/10.1145/130226.134466}
  {\path{doi:10.1145/130226.134466}}.

\bibitem{Charikar2004}
Moses Charikar and Rina Panigrahy.
\newblock Clustering to minimize the sum of cluster diameters.
\newblock {\em Journal of Computer and System Sciences}, 68(2):417 -- 441,
  2004.
\newblock \href
  {http://dx.doi.org/http://dx.doi.org/10.1016/j.jcss.2003.07.014}
  {\path{doi:http://dx.doi.org/10.1016/j.jcss.2003.07.014}}.

\bibitem{Csirik2013}
J{\'{a}}nos Csirik, Leah Epstein, Csan{\'{a}}d Imreh, and Asaf Levin.
\newblock Online clustering with variable sized clusters.
\newblock {\em Algorithmica}, 65(2):251--274, 2013.
\newblock \href {http://dx.doi.org/10.1007/s00453-011-9586-2}
  {\path{doi:10.1007/s00453-011-9586-2}}.

\bibitem{Doddi2000}
Srinivas~R. Doddi, Madhav~V. Marathe, Sekharipuram~S. Ravi, David~S. Taylor,
  and Peter Widmayer.
\newblock {A}pproximation {A}lgorithms for {C}lustering to {M}inimize the {S}um
  of {D}iameters.
\newblock {\em Nordic journal of computing}, 7(3):185--203, 2000.

\bibitem{Fotakis2011}
Dimitris Fotakis and Paraschos Koutris.
\newblock Online sum-radii clustering.
\newblock {\em CoRR}, abs/1109.5325, 2011.
\newblock URL: \url{http://arxiv.org/abs/1109.5325}.

\bibitem{gibson2008}
Matt Gibson, Gaurav Kanade, Erik Krohn, Imran~A. Pirwani, and Kasturi
  Varadarajan.
\newblock On clustering to minimize the sum of radii.
\newblock In {\em Proceedings of the Nineteenth Annual ACM-SIAM Symposium on
  Discrete Algorithms}, SODA '08, pages 819--825, Philadelphia, PA, USA, 2008.
  Society for Industrial and Applied Mathematics.
\newblock URL: \url{http://dl.acm.org/citation.cfm?id=1347082.1347172}.

\bibitem{Gibson2008b}
Matt Gibson, Gaurav Kanade, Erik Krohn, Imran~A. Pirwani, and Kasturi
  Varadarajan.
\newblock On metric clustering to minimize the sum of radii.
\newblock In {\em Proceedings of the 11th Scandinavian Workshop on Algorithm
  Theory}, SWAT '08, pages 282--293, Berlin, Heidelberg, 2008. Springer-Verlag.
\newblock \href {http://dx.doi.org/10.1007/978-3-540-69903-3_26}
  {\path{doi:10.1007/978-3-540-69903-3_26}}.

\bibitem{gupta2016}
Anupam Gupta, Ravishankar Krishnaswamy, Amit Kumar, and Debmalya Panigrahi.
\newblock Online and dynamic algorithms for set cover.
\newblock {\em CoRR}, abs/1611.05646, 2016.
\newblock URL: \url{http://arxiv.org/abs/1611.05646}.

\bibitem{hansen1987}
P.~Hansen and B.~Jaumard.
\newblock Minimum sum of diameters clustering.
\newblock {\em Journal of Classification}, 4(2):215--226, 1987.

\bibitem{Hansen1997}
Pierre Hansen and Brigitte Jaumard.
\newblock Cluster analysis and mathematical programming.
\newblock {\em Math. Program.}, 79(1-3):191--215, October 1997.
\newblock \href {http://dx.doi.org/10.1007/BF02614317}
  {\path{doi:10.1007/BF02614317}}.

\bibitem{Krauthgamer2004}
Robert Krauthgamer and James~R. Lee.
\newblock Navigating nets: Simple algorithms for proximity search.
\newblock In {\em Proceedings of the Fifteenth Annual ACM-SIAM Symposium on
  Discrete Algorithms}, SODA '04, pages 798--807, Philadelphia, PA, USA, 2004.
  Society for Industrial and Applied Mathematics.
\newblock URL: \url{http://dl.acm.org/citation.cfm?id=982792.982913}.

\bibitem{Lev-Tov2005}
Nissan Lev-Tov and David Peleg.
\newblock Polynomial time approximation schemes for base station coverage with
  minimum total radii.
\newblock {\em Comput. Netw. ISDN Syst.}, 47(4):489--501, March 2005.
\newblock \href {http://dx.doi.org/10.1016/j.comnet.2004.08.012}
  {\path{doi:10.1016/j.comnet.2004.08.012}}.

\bibitem{proietti2006}
Guido Proietti and Peter Widmayer.
\newblock Partitioning the nodes of a graph to minimize the sum of subgraph
  radii.
\newblock In {\em Proceedings of the 17th International Conference on
  Algorithms and Computation}, ISAAC'06, pages 578--587, Berlin, Heidelberg,
  2006. Springer-Verlag.
\newblock \href {http://dx.doi.org/10.1007/11940128_58}
  {\path{doi:10.1007/11940128_58}}.

\bibitem{schaeffer2007}
Satu~Elisa Schaeffer.
\newblock Survey: Graph clustering.
\newblock {\em Comput. Sci. Rev.}, 1(1):27--64, August 2007.
\newblock \href {http://dx.doi.org/10.1016/j.cosrev.2007.05.001}
  {\path{doi:10.1016/j.cosrev.2007.05.001}}.

\bibitem{Solomon2016}
Shay Solomon.
\newblock Fully dynamic maximal matching in constant update time.
\newblock {\em IEEE FOCS}, 2016.
\newblock URL: \url{http://arxiv.org/abs/1604.08491}.

\end{thebibliography}

\end{document}